\documentclass{symmetry10}

\begin{document}

\FirstPageHeading{Bihlo&Popovych}

\ShortArticleName{Point symmetry group of the barotropic vorticity equation}

\ArticleName{Point Symmetry Group\\ of the Barotropic Vorticity Equation}

\Author{Alexander BIHLO~$^\dag$ and Roman O.~POPOVYCH~$^{\dag\ddag}$}
\AuthorNameForHeading{A.~Bihlo and R.O.~Popovych}
\AuthorNameForContents{Bihlo A.\ and POPOVYCH R.O.}
\ArticleNameForContents{Point symmetry group of the barotropic vorticity equation}

\Address{$^\dag$~Fakult\"at f\"ur Mathematik, Universit\"at Wien,\\
\hphantom{$^\ddag$}~Nordbergstra{\ss}e 15, A-1090 Wien, Austria}
\EmailD{alexander.bihlo@univie.ac.at, rop@imath.kiev.ua}

\Address{$^\ddag$~Institute of Mathematics of NAS of Ukraine,\\
\hphantom{$^\dag$}~3 Tereshchenkivska Str., Kyiv-4, Ukraine}

\Abstract{The complete point symmetry group of the barotropic vorticity equation on the $\beta$-plane is
computed using the direct method supplemented with two different techniques.
The first technique is based on the preservation of any megaideal of the maximal Lie invariance algebra
of a differential equation by the push-forwards of point symmetries of the same equation.
The second technique involves a priori knowledge on normalization properties of a class of differential equations
containing the equation under consideration.
Both of these techniques are briefly outlined.}

\section{Introduction}

It is well known that it is much easier to determine the continuous part of the complete point symmetry group of a differential equation than the entire group including discrete symmetries. The computation of continuous (Lie) symmetries is possible using infinitesimal techniques, which amounts to solving an overdetermined system of linear partial differential equations (referred to as \emph{determining equations}) for coefficients of vector fields generating one-parameter Lie symmetry groups. Owing to the algorithmic nature of this problem, the automatic computation of Lie symmetries is already implemented in a number of symbolic calculation packages, see, e.g., papers~\cite{Bihlo&Popovych:Carminati&Khai2000,Bihlo&Popovych:Head1993,Bihlo&Popovych:RochaFilho&Figueiredo2010}
for detail description of certain packages and
reviews~\cite{Bihlo&Popovych:Hereman1997,Bihlo&Popovych:Butcher&Carminati&Vu2003}.

The relative simplicity of finding Lie symmetries of differential equations is also a primary reason why the overwhelming part of research on symmetries is devoted to symmetries of this kind. See, e.g., the textbooks \cite{Bihlo&Popovych:Bluman&Cheviakov&Anco2010,Bihlo&Popovych:Bluman&Kumei1989,
Bihlo&Popovych:Meleshko2005,Bihlo&Popovych:Olver2000,Bihlo&Popovych:Ovsiannikov1982} for general theory and numerous examples and additionally the works \cite{Bihlo&Popovych:Andreev&Kaptsov&Pukhnachov&Rodionov1998,Bihlo&Popovych:Bihlo&Popovych2009a, Bihlo&Popovych:Bihlo&Popovych2009b,Bihlo&Popovych:Fushchych&Popovych1994,Bihlo&Popovych:Meleshko2004} for several applications of Lie methods in hydrodynamics and meteorology.

As continuous symmetries, also discrete symmetries are of practical relevance in a number of fields such as dynamical system theory, quantum mechanics, crystallography and solid state physics. They can also be helpful in some issues related to Lie symmetries, e.g.\ allowing for a simplification of optimal lists of inequivalent subalgebras, and due to enabling the construction of new solutions of differential equations from known ones. It is not possible, in general, to determine the whole point symmetry group in terms of finite transformations by usage of infinitesimal techniques. On the other hand, the direct computation of point symmetries based on their definition boils down to solving a cumbersome nonlinear system of determining equations, which is difficult to be integrated. Similar determining equations also arise under calculations of equivalence groups and sets of admissible transformations of classes of differential equations by means of employing the direct method. In order to simplify the derivation of the determining equations, different special techniques have been developed involving, in particular, the implicit representation of unknown functions, the combined splitting with respect to old and new variables and the inverse expression of old derivative via new ones~\cite{Bihlo&Popovych:Popovych&Bihlo2010,Bihlo&Popovych:Popovych&Kunzinger&Eshraghi2010,Bihlo&Popovych:Prokhorova2005}.

There exist two particular techniques that can be applied for \emph{a priori} simplification of calculations concerning the point symmetry groups of differential equations.

The first technique was presented in~\cite{Bihlo&Popovych:Hydon2000} for equations whose maximal Lie invariance algebras are finite dimensional.
It is based on the fact that the push-forwards of point symmetries of a given system of differential equations to vector fields on the space of dependent and independent variables are automorphisms of the maximal Lie invariance algebra of the same system.
This condition yields restrictions for those point transformations that can qualify as symmetries of the system of differential equations under consideration. We will adopt this technique to the infinite dimensional case using the notion of megaideals of Lie algebras, which are the most invariant algebraic structures.

The second technique involves available information on the set of admissible transformations of a class of differential equations~\cite{Bihlo&Popovych:Popovych&Kunzinger&Eshraghi2010}, which contains the investigated equation.

In the present paper, we will demonstrate both of these techniques by computing the complete point symmetry group of the barotropic vorticity equation on the $\beta$-plane. This is one of the most classical models which are used in geophysical fluid dynamics. The techniques to be employed are briefly described in Section~\ref{Bihlo&Popovych:sec:Techniques}. The actual computations using the method based on the corresponding Lie invariance algebra and that involving a priori knowledge on admissible transformations of a class of generalized vorticity equations are presented in Section~\ref{Bihlo&Popovych:sec:CalculationsInvarianceAlgebra} and~\ref{Bihlo&Popovych:sec:DirectMethodForBVE}, respectively. A short summary concludes the paper. %~\ref{sec:Conclusion}.

\section{Techniques of calculation \\ of complete point symmetry groups}\label{Bihlo&Popovych:sec:Techniques}

Both the techniques described in this section should be considered merely as tools for deriving preliminary restrictions on point symmetries.
In either case, calculations must finally be carried out within the framework of the direct approach.

\subsection{Using megaideals of Lie invariance algebra}\label{Bihlo&Popovych:sec:TechniquesLieInvarianceAlgebra}

The most refined version of the technique involving Lie symmetries in the calculations of complete point symmetry groups was applied in~\cite{Bihlo&Popovych:Hydon2000}.
It is outlined as follows:
Given a system of differential equations~$\mathcal L$ whose maximal Lie invariance algebra $\mathfrak g$ is $n$-dimensional with a basis $\{e_1,\dots,e_n\}$, $n<\infty$, one has to compute the entire automorphism group of $\mathfrak g$, $\mathrm{Aut}(\mathfrak g)$. Supposing that $\mathcal T$ is a transformation from the complete point symmetry group~$G$ of~$\mathcal L$, one has the condition $\mathcal T_* e_j=\sum_{i=1}^n e_ia_{ij}$ for $j=1,\dots,n$, where $\mathcal T_*$ denotes the push-forward of vector fields induced by $\mathcal T$ and $(a_{ij})$ is the matrix of an automorphism of $\mathfrak g$ in the chosen basis. This condition implies constraints on the transformation $\mathcal T$ which are then taken into account in further calculations with the direct method.

The method we propose here is different to those described in the previous paragraph. In fact, it uses only the minimal information on the automorphism group $\mathrm{Aut}(\mathfrak g)$ in the form of a set of megaideals of $\mathfrak g$. Due to this, it is applicable also in the case when the maximal Lie invariance algebra is infinite dimensional. The notion of megaideals was introduced in~\cite{Bihlo&Popovych:Popovych&Boyko&Nesterenko&Lutfullin2005}.
\begin{definition}
A \emph{megaideal} $\mathfrak i$ is a vector subspace of $\mathfrak g$ that is invariant under any transformation from the automorphism group $\mathrm{Aut}(\mathfrak g)$ of $\mathfrak g$.
\end{definition}
That is, we have $\mathfrak T \mathfrak i=\mathfrak i$ for a megaideal~$\mathfrak i$ of~$\mathfrak g$, whenever $\mathfrak T$ is a transformation from $\mathrm{Aut}(\mathfrak g)$. Any megaideal of~$\mathfrak g$ is an ideal and characteristic ideal of~$\mathfrak g$.
Both the improper subalgebras of~$\mathfrak g$ (the zero subspace and $\mathfrak g$ itself) are megaideals of~$\mathfrak g$.
The following assertions are obvious.

\begin{proposition}\label{prop:OnMegaIdeals1}
If $\mathfrak i_1$ and $\mathfrak i_2$ are megaideals of~$\mathfrak g$ then so are $\mathfrak i_1+\mathfrak i_2,$ $\mathfrak i_1\cap \mathfrak i_2$ and $[\mathfrak i_1,\mathfrak i_2]$, i.e., sums, intersections and Lie products of megaideals are again megaideals.
\end{proposition}

\begin{proposition}\label{prop:OnMegaIdeals2}
If $\mathfrak i_2$ is a megaideal of $\mathfrak i_1$ and $\mathfrak i_1$ is a megaideal of $\mathfrak g$ then $\mathfrak i_2$ is a megaideal of $\mathfrak g$, i.e., megaideals of megaideals are also megaideals.
\end{proposition}

\begin{corollary}\label{cor:OnMegaIdeals3}
All elements of the derived, upper and lower central series of a Lie algebra are its megaideals. In particular, the center and the derivative of a Lie algebra are its megaideals.
\end{corollary}

\begin{corollary}\label{cor:OnMegaIdeals4}\looseness=-1
The radical~$\mathfrak r$ and nil-radical~$\mathfrak n$ (i.e., the maximal solvable and nilpotent ideals, respectively) of~$\mathfrak g$
as well as different Lie products, sums and intersections involving~$\mathfrak g$, $\mathfrak r$ and~$\mathfrak n$
($[\mathfrak g,\mathfrak r]$, $[\mathfrak r,\mathfrak r]$, $[\mathfrak g,\mathfrak n]$, $[\mathfrak r,\mathfrak n]$,  $[\mathfrak n,\mathfrak n]$, etc.) are megaideals of~$\mathfrak g$.
\end{corollary}

Suppose that $\mathfrak g$ is finite dimensional and possesses a megaideal $\mathfrak i$ which, without loss of generality, can be assumed to be spanned by the first $k$ basis elements, $\mathfrak i=\langle e_1,\dots,e_k\rangle$. Then the matrix $(a_{ij})$ of any automorphism of $\mathfrak g$ has block structure, namely, $a_{ij}=0$ for $i>k$. In other words, in the finite dimensional case we take into account only the block structure of automorphism matrices. This is reasonable as the entire automorphism group $\mathrm{Aut}(\mathfrak g)$ (which should be computed within the method from~\cite{Bihlo&Popovych:Hydon2000}) may be much wider than the group of automorphisms of $\mathfrak g$ induced by elements of the point symmetry group~$G$ of~$\mathcal L$. Moreover, it seems difficult to find the entire group $\mathrm{Aut}(\mathfrak g)$ if the algebra~$\mathfrak g$ is infinite dimensional. At the same time, in view of the above assertions it is easy to determine a set of megaideals for any Lie algebra.

\subsection{Direct method and admissible transformations}\label{Bihlo&Popovych:sec:DirectMethodAndAdmTrans}

The initial point of the second technique is to consider a given $p$th order system~$\mathcal L^0$ of $l$~differential equations
for $m$~unknown functions $u=(u^1,\ldots,u^m)$ of $n$~independent variables $x=(x_1,\ldots,x_n)$
as an element of a class~$\mathcal L|_{\mathcal S}$ of similar systems~$\mathcal L_\theta$: $\smash{L(x,u_{(p)},\theta(x,u_{(p)}))=0}$
parameterized by a tuple of $p$th order differential functions (arbitrary elements)~$\theta=(\theta^1(x,u_{(p)}),\ldots,\theta^k(x,u_{(p)}))$.
Here $u_{(p)}$ denotes the set of all the derivatives of~$u$ with respect to $x$
of order not greater than~$p$, including $u$ as the derivatives of order zero.
The class~$\mathcal L|_{\mathcal S}$ is determined by two objects:
the tuple $L=(L^1,\ldots,L^l)$ of $l$ fixed functions depending on~$x$, $u_{(p)}$ and~$\theta$ and~$\theta$ running through the set~$\mathcal S$.
Within the framework of symmetry analysis of differential equations,
the set~$\mathcal S$ is defined as the set of solutions of an auxiliary system consisting of
a subsystem $S(x,u_{(p)},\theta_{(q)}(x,u_{(p)}))=0$ of differential equations with respect to $\theta$
and a non-vanish condition $\Sigma(x,u_{(p)},\theta_{(q)}(x,u_{(p)}))\ne0$
with another differential function $\Sigma$ of~$\theta$.
In the auxiliary system, $x$ and $u_{(p)}$ play the role of independent variables
and $\theta_{(q)}$ stands for the set of all the partial derivatives of $\theta$ of order not greater than $q$
with respect to the variables $x$ and $u_{(p)}$.
In view of the purpose of our consideration we should have that $\mathcal L^0=\mathcal L_{\theta_0}$ for some $\theta_0\in\mathcal S$.

Following~\cite{Bihlo&Popovych:Popovych&Kunzinger&Eshraghi2010},
for $\theta,\tilde\theta\in\mathcal S$ we denote by $\mathrm T(\theta,\tilde\theta)$ the set of point transformations
which map the system~$\mathcal L_\theta$ to the system~$\mathcal L_{\tilde\theta}$.
The maximal point symmetry group~$G_\theta$ of the system~$\mathcal L_\theta$ coincides with~$\mathrm T(\theta,\theta)$.

\begin{definition}\label{DefOfSetOfAdmTrans}
$\mathrm T(\mathcal L|_{\mathcal S})=\{(\theta,\tilde\theta,\varphi)\mid
\theta,\tilde\theta\in\mathcal S,\,%\mathrm T(\theta,\tilde\theta)\not=\varnothing,\,
\varphi\in\mathrm T(\theta,\tilde\theta)\}$
is called the {\em set of admissible transformations in~$\mathcal L|_{\mathcal S}$}.
\end{definition}

Sets of admissible transformations were first systematically described by King\-ston and Sophocleous
for a class of generalized Burgers equations~\cite{Bihlo&Popovych:Kingston&Sophocleous1991}
and Winternitz and Gazeau for a class of variable coefficient Korteweg--de Vries equations~\cite{Bihlo&Popovych:Winternitz&Gazeau1992},
in terms of {\em form-preserving}
\cite{Bihlo&Popovych:Kingston&Sophocleous1991,Bihlo&Popovych:Kingston&Sophocleous1998,Bihlo&Popovych:Kingston&Sophocleous2001}
and {\em allowed}~\cite{Bihlo&Popovych:Winternitz&Gazeau1992} transformations, respectively.
The notion of admissible transformations can be considered as a formalization of their approaches.

Any point symmetry transformation of an equation~$\mathcal L_\theta$ from the class~$\mathcal L|_{\mathcal S}$ generates an admissible transformation in this class.
Therefore, it obviously satisfies all restrictions which hold for admissible transformations \cite{Bihlo&Popovych:Kingston&Sophocleous1998}.
For example, it is known for a long time that for any point (and even contact) transformation connecting a pair of $(1+1)$-dimensional evolution equations
its component corresponding to~$t$ depends only on~$t$, cf.~\cite{Bihlo&Popovych:Magadeev1993}.
The equations in the pair can also coincide.
As a result, the same restriction should be satisfied by any point or contact symmetry transformation of every $(1+1)$-dimensional evolution equation.

The simplest description of admissible transformations is obtained for normalized classes of differential equations.
Roughly speaking, a class of (systems of) differential equations is called \emph{normalized} if any admissible transformation in this class is induced
by a transformation from its equivalence group.
Different kinds of normalization can be defined depending on what kind of equivalence group
(point, contact, usual, generalized, extended, etc.) is considered.
Thus, the \emph{usual equivalence group}~$G^{\sim}$ of the class~$\mathcal L|_{\mathcal S}$
consists of those point transformations in the space of variables and arbitrary elements,
which are projectable on the variable space and preserve the whole class~$\mathcal L|_{\mathcal S}$.
The class~$\mathcal L|_{\mathcal S}$ is called normalized in the usual sense if
the set $\mathrm T(\mathcal L|_{\mathcal S})$ is generated by the usual equivalence group~$G^{\sim}$.
As a consequence, all generalizations of the equivalence group within the framework of point transformations are trivial for this class.
See~\cite{Bihlo&Popovych:Popovych&Kunzinger&Eshraghi2010} for precise definitions and further explanations.
If the class~$\mathcal L|_{\mathcal S}$ is normalized in certain sense with respect to point transformations,
the point symmetry group~$G_{\theta_0}$ of any equation~$\mathcal L_{\theta_0}$ from this class
is contained in the projection of the corresponding equivalence group of~$\mathcal L|_{\mathcal S}$
to the space of independent and dependent variables (taken for the value $\theta=\theta_0$ in the case when the generalized equivalence group is considered).

As a rule, calculations of certain common restrictions on admissible transformations of the entire normalized class or
its normalized subclasses or point symmetry transformations of a single equation from this class have the same level of complexity.
For example, in order to derive the restriction that the transformation component corresponding to~$t$ depends only on~$t$,
we should carry out approximately the same operations, independently of considering the whole class of $(1+1)$-dimensional evolution equations,
any well-defined subclass from this class or any single evolution equation.
This is why it is worthwhile to first construct nested series of normalized classes of differential equations by
starting from a quite general, obviously normalized class, imposing on each step additional auxiliary conditions on the arbitrary elements
and then studying the complete point symmetries of a single equation from the narrowest class of the constructed series.

In the way outlined above we have already investigated hierarchies of normalized classes of
generalized nonlinear Schr\"odinger equations~\cite{Bihlo&Popovych:Popovych&Kunzinger&Eshraghi2010},
$(1+1)$-dimensional linear evolution equations~\cite{Bihlo&Popovych:Popovych&Kunzinger&Ivanova2008},
$(1+1)$-dimensional third-order evolution equations including variable-coefficient Korteweg--de Vries and modified Korteweg--de Vries equations~\cite{Bihlo&Popovych:Popovych&Vaneeva2010} and
generalized vorticity equations arising in the study of local parameterization schemes for the barotropic vorticity equation~\cite{Bihlo&Popovych:Popovych&Bihlo2010}.

If an equation does not belong to a class whose admissible transformations have been studied earlier,
one can try to map this equation using a point transformation to an equation from a class for which constraints on its admissible transformations are known a priori.
Then one can either map the known constraints on admissible transformations back and then complete the calculations of point symmetries of the initial equation using the direct method or calculate the point symmetry group of the mapped equation using the direct method and then map this group back.
The example on the application of this trick to the barotropic vorticity equation in presented in Section~\ref{Bihlo&Popovych:sec:DirectMethodForBVE}.

\section{Calculations based on Lie invariance algebra\\ of the barotropic vorticity equation}\label{Bihlo&Popovych:sec:CalculationsInvarianceAlgebra}

The barotropic vorticity equation on the $\beta$-plane reads
\begin{align}\label{Bihlo&Popovych:eq:vortbeta}
    \zeta_t+\psi_x\zeta_y-\psi_y\zeta_x+\beta\psi_x=0,
\end{align}
where $\psi=\psi(t,x,y)$ is the stream function and $\zeta:=\psi_{xx}+\psi_{yy}$ is the relative vorticity, which is the vertical component of the vorticity vector. The barotropic vorticity equation in the formulation~\eqref{Bihlo&Popovych:eq:vortbeta} is valid in situations where the two-dimensional wind field can be regarded as almost non-divergent and the motion in North--South direction is confined to a relatively small region. It is then convenient to use a local Cartesian coordinate system. In such a coordinate system, the effect of the sphericity of the Earth is conveniently taken into account by approximating the normal component of the vorticity due to the rotation of the Earth, $2\Omega\sin\varphi$, by its linear Taylor series expansion, where $\Omega$ is the angular rotation of the Earth and $\varphi$ is the geographic latitude. This linear approximation at some reference latitude $\varphi_0$ is given by $2\Omega\sin\varphi_0+\beta y$, where $\beta=2\Omega\cos\varphi/a$ and $a$ is the radius of the Earth. This is the traditional $\beta$-plane approximation, see~\cite{Bihlo&Popovych:Pedlosky1987} for further details. Then, taking the vertical component of the curl of the two-dimensional ideal Euler equations and using the $\beta$-plane approximation leads to Eq.~\eqref{Bihlo&Popovych:eq:vortbeta}.

It is straightforward to determine the maximal Lie invariance algebra~$\mathfrak g$ of Eq.~\eqref{Bihlo&Popovych:eq:vortbeta} using infinitesimal techniques:
\[
 \mathfrak g=\langle\mathcal{D},\partial_t,\partial_y,\mathcal{X}(f),\mathcal{Z}(g)\rangle,
\]
where $\mathcal{D}=t\partial_t-x\partial_x-y\partial_y-3\psi\partial_\psi$, $\mathcal{X}(f)=f(t)\partial_x-f_t(t)y\partial_\psi$ and \mbox{$\mathcal{Z}(g)=g(t)\partial_\psi$},
and $f$ and $g$ run through the space of smooth functions of $t$.
(In fact, the precise interpretation of~$\mathfrak g$ as a Lie algebra strongly depends on what space of smooth functions is chosen for~$f$ and~$g$, cf.\ Note~A.1 in \cite[p.~178]{Bihlo&Popovych:Fushchych&Popovych1994}.)
This result was first obtained in~\cite{Bihlo&Popovych:Katkov1965} and is now easily accessible in the handbook \cite[p.~223]{Bihlo&Popovych:Ibragimov1995}.
See also~\cite{Bihlo&Popovych:Bihlo&Popovych2009a} for related discussions and the exhaustive study of the classical Lie reductions of Eq.~\eqref{Bihlo&Popovych:eq:vortbeta}.

The nonzero commutation relations of the algebra $\mathfrak g$ in the above basis are exhausted by the following ones:
\begin{gather*}
[\partial_t,\mathcal{D}]=\partial_t,\quad [\partial_y,\mathcal{D}]=-\partial_y,\\
[\mathcal{D},\mathcal{X}(f)]=\mathcal{X}(tf_t+f),\quad [\mathcal{D},\mathcal{Z}(g)]=\mathcal{Z}(tg_t+3g),\\
[\partial_t,\mathcal{X}(f)]=\mathcal{X}(f_t),\quad [\partial_t,\mathcal{Z}(g)]=\mathcal{Z}(g_t),\quad [\partial_y,\mathcal{X}(f)]=-\mathcal{Z}(f_t).
\end{gather*}
It is easy to see from the commutation relations that the Lie algebra $\mathfrak g$ is solvable since
\begin{gather*}
\mathfrak g'=[\mathfrak g,\mathfrak g]=\langle\partial_t,\partial_y,\mathcal{X}(f),\mathcal{Z}(g)\rangle, \\
\mathfrak g''=[\mathfrak g',\mathfrak g']=\langle\mathcal{X}(f),\mathcal{Z}(g)\rangle,\\
\mathfrak g'''=[\mathfrak g'',\mathfrak g'']=0.
\end{gather*}
Therefore, the radical~$\mathfrak r$ of~$\mathfrak g$ coincides with the entire algebra~$\mathfrak g$.
The nil-radical of~$\mathfrak g$ is the ideal
\[
\mathfrak n=\langle\partial_y,\mathcal{X}(f),\mathcal{Z}(g)\rangle.
\]
Indeed, this ideal is a nilpotent subalgebra of~$\mathfrak g$ since
\[
\mathfrak n^{(2)}=\mathfrak n'=[\mathfrak n,\mathfrak n]=\langle\mathcal{Z}(g)\rangle, \quad
\mathfrak n^{(3)}=[\mathfrak n,\mathfrak n']=0.
\]
It can be extended to a larger ideal of~$\mathfrak g$ only with two sets of elements, $\{\partial_t\}$ and $\{\mathcal{D},\partial_t\}$.
Both resulting ideals are not nilpotent.
In other words, $\mathfrak n$ is the maximal nilpotent ideal.

Continuous point symmetries of Eq.~\eqref{Bihlo&Popovych:eq:vortbeta} are determined from the elements of $\mathfrak g$ by integration of the associated Cauchy problems. It is obvious that Eq.~\eqref{Bihlo&Popovych:eq:vortbeta} also possesses two discrete symmetries, $(t,x,y,\psi)\mapsto (-t,-x,y,\psi)$ and $(t,x,y,\psi)\mapsto (t,x,-y,-\psi)$, which are independent up to their composition and their compositions with continuous symmetries. The proof that the above symmetries generate the entire point symmetry group was, however, outstanding.

\begin{theorem}\label{Bihlo&Popovych:TheoremOnPointSymGroupOfBVE}
 The complete point symmetry group of the barotropic vorticity equation on the $\beta$-plane~\eqref{Bihlo&Popovych:eq:vortbeta} is formed by the transformations
\begin{align*}
\mathcal T\colon &\quad \tilde t=T_1t+T_0, \quad \tilde x=\frac{1}{T_1}x+f(t), \quad \tilde y=\frac{\varepsilon}{T_1}y+Y_0, \\
  &\quad \tilde\psi=\frac{\varepsilon}{(T_1)^3}\psi-\frac{\varepsilon}{(T_1)^2}f_t(t)y+g(t),
\end{align*}
where $T_1\ne0$, $\varepsilon=\pm1$ and $f$ and $g$ are arbitrary functions of $t$.
\end{theorem}

\begin{proof}
The discrete symmetries of the barotropic vorticity equation on the $\beta$-plane are computed as described in section~\ref{Bihlo&Popovych:sec:TechniquesLieInvarianceAlgebra}. The general form of a point transformation of the vorticity equation is:
\[
 \mathcal T\colon\quad (\tilde t, \tilde x,\tilde y, \tilde \psi)=(T, X, Y, \Psi),
\]
where $T$, $X$, $Y$ and $\Psi$ are regarded as functions of $t$, $x$, $y$ and $\psi$, whose joint Jacobian does not vanish.
To obtain the constrained form of $\mathcal T$, we use the above four proper nested megaideals of~$\mathfrak g$,
namely $\mathfrak n'$, $\mathfrak g''$, $\mathfrak n$ and $\mathfrak g'$, and~$\mathfrak g$ itself.
Recall once more that the transformation $\mathcal T$ must satisfy the conditions
$\mathcal T_* \mathfrak n'=\mathfrak n'$,
$\mathcal T_* \mathfrak g''=\mathfrak g''$,
$\mathcal T_* \mathfrak n=\mathfrak n$,
$\mathcal T_* \mathfrak g'=\mathfrak g'$ and
$\mathcal T_* \mathfrak g=\mathfrak g$
in order to qualify as a point symmetry of the vorticity equation, where $\mathcal T_*$ denotes the push-forward of $\mathcal T$ to vector fields. In other words, we have
\begin{gather}
\mathcal T_* \mathcal{Z}(g)=g(T_\psi\partial_{\tilde t}+X_\psi\partial_{\tilde x}+Y_\psi \partial_{\tilde y}+\Psi_\psi \partial_{\tilde \psi})=\mathcal{\tilde Z} (\tilde g^g),
\label{Bihlo&Popovych:eq:MegaidealConstraintForT1}\\
\mathcal T_* \mathcal{X}(f)=\mathcal{\tilde X}(\tilde f^f)+\mathcal{\tilde Z}(\tilde g^f),
\label{Bihlo&Popovych:eq:MegaidealConstraintForT2}\\
\mathcal T_* \partial_t=T_t\partial_{\tilde t}+X_t\partial_{\tilde x}+Y_t\partial_{\tilde y}+\Psi_t\partial_{\tilde \psi}=a_1\partial_{\tilde t}+a_2\partial_{\tilde y}+\mathcal{\tilde X}(\tilde f)+\mathcal{\tilde Z}(\tilde g),
\label{Bihlo&Popovych:eq:MegaidealConstraintForT3}\\
\mathcal T_* \partial_y=T_y\partial_{\tilde t}+X_y\partial_{\tilde x}+Y_y\partial_{\tilde y}+\Psi_y\partial_{\tilde \psi}=b_1\partial_{\tilde y}+\mathcal{\tilde X}(\tilde f^y)+\mathcal{\tilde Z}(\tilde g^y),
\label{Bihlo&Popovych:eq:MegaidealConstraintForT4}\\
\mathcal T_* \mathcal{D}=c_1\mathcal{\tilde D}+c_2\partial_{\tilde t}+c_3\partial_{\tilde y}+\mathcal{\tilde X}(\tilde f^D)+\mathcal{\tilde Z}(\tilde g^D),
\label{Bihlo&Popovych:eq:MegaidealConstraintForT5}
\end{gather}
where all $\tilde f$'s and $\tilde g$'s are smooth functions of~$\tilde t$ which are determined,
as the constant parameters $a_1$, $a_2$, $b_1$, $c_1$, $c_2$ and~$c_3$, by~$\mathcal T_*$ and the operator from the corresponding left-hand side.

We will derive constraints on~$\mathcal T_*$, consequently equating coefficients of vector fields in conditions \eqref{Bihlo&Popovych:eq:MegaidealConstraintForT1}--\eqref{Bihlo&Popovych:eq:MegaidealConstraintForT5} and taking into account constraints obtained on previous steps.
Thus Eq.~\eqref{Bihlo&Popovych:eq:MegaidealConstraintForT1} immediately implies $T_\psi=X_\psi=Y_\psi=0$ (hence $\Psi_\psi\ne0$) and $g\Psi_\psi=\tilde g^g$.
Evaluating the last equation for $g=1$ and $g=t$ and combining the results gives $t=\tilde g^t(T)/\tilde g^1(T)$, where $\tilde g^t=\tilde g^g|_{g=t}$ and $\tilde g^1=\tilde g^g|_{g=1}$.
As the derivative with respect to $T$ in the right hand side of this equality does not vanish, the condition $T=T(t)$ must hold. This implies that $\Psi_\psi$ depends only on~$t$.

As then $\mathcal T_*\mathcal{X}(f)=fX_x\partial_{\tilde x}+fY_x\partial_{\tilde y}+(f\Psi_x-f_ty\Psi_\psi)\partial_{\tilde \psi}$,
it follows from Eq.~\eqref{Bihlo&Popovych:eq:MegaidealConstraintForT2} that $Y_x=0$ and
\[
fX_x=\tilde f^f,\quad f\Psi_x-f_ty\Psi_\psi=-\tilde f^f_{\tilde t}Y+\tilde g^f.
\]
\looseness=-1
Evaluating the first of the displayed equalities for $f=1$, we derive that $X_x=\tilde f^1(T)=: X_1(t)$.
Therefore, $\tilde f^f(T)=f(t)X_1(t)$. The second equality then reads
\[
 f\Psi_x-f_ty\Psi_\psi=-\frac{(fX_1)_t}{T_t}Y+\tilde g^f.
\]
Setting $f=1$ and $f=t$ in the last equality and combining the resulting equalities yields $y\Psi_\psi=(T_t)^{-1}X_1Y+t\tilde g^1-\tilde g^t$, where $\tilde g^t=\tilde g^f|_{f=t}$ and $\tilde g^1=\tilde g^f|_{f=1}$. As $X_1\ne0$ this equation implies that $Y=Y_1(t)y+Y_0(t)$.

After analyzing Eq.~\eqref{Bihlo&Popovych:eq:MegaidealConstraintForT3}, we find
$T_t=\mathop{\rm const}\nolimits$, $Y_t=\mathop{\rm const}\nolimits$, which leads to $Y_1=\mathop{\rm const}\nolimits$, $X_t=\tilde f(T)$ and thus $X_{tx}=0$, i.e., $X_1=\mathop{\rm const}\nolimits$. Finally, Eq.~\eqref{Bihlo&Popovych:eq:MegaidealConstraintForT3} also implies $\Psi_t=-\tilde f_{\tilde t}Y+\tilde g$.
In a similar manner, upon taking into account the restrictions already derived so far, collecting coefficients in Eq.~\eqref{Bihlo&Popovych:eq:MegaidealConstraintForT4} gives the constraint $X_y=\tilde f^y=:X_2=\mathop{\rm const}\nolimits$ since $X_{yt}=0$. Moreover, $\Psi_y=\tilde g^y$, as $\tilde f^y_{\tilde t}=0$.

The final restrictions on $\mathcal T$ based on the preservation of~$\mathfrak g$ are derivable from Eq.~\eqref{Bihlo&Popovych:eq:MegaidealConstraintForT5}, where
\begin{align*}
\mathcal T_* \mathcal{D}={}&
tT_t\partial_{\tilde t}+(tX_t-xX_x-yX_y)\partial_{\tilde x}+(tY_t-yY_y)\partial_{\tilde y}\\
&\!\!+(t\Psi_t-x\Psi_x-y\Psi_y-3\psi\Psi_\psi)\partial_{\tilde \psi}.
\end{align*}
Collecting the coefficients of $\partial_{\tilde t}$ and $\partial_{\tilde y}$, we obtain that $c_1=1$ and $Y_t=0$. Similarly, equating the coefficients of $\partial_{\tilde\psi}$ and further splitting with respect to $x$ implies that $\Psi_x=0$.

The results obtained so far lead to the following constrained form of the general point symmetry transformation of the vorticity equation~\eqref{Bihlo&Popovych:eq:vortbeta}
\begin{gather}\label{Bihlo&Popovych:eq:RestrictedFormOfPointTransForBVE}
\begin{split}
&T=T_1t+T_0, \quad X=X_1x+X_2y+f(t), \quad Y=Y_1y+Y_0, \\
&\Psi=\Psi_1\psi+\Psi_2(t)y+\Psi_4(t),
\end{split}
\end{gather}
where $T_0$, $T_1$, $X_1$, $X_2$, $Y_0$, $Y_1$ and $\Psi_1$ are arbitrary constants, $T_1X_1Y_1\Psi_1\ne0$, and $f(t)$, $\Psi_2(t)$ and $\Psi_4(t)$ are arbitrary time-dependent functions.
The form~\eqref{Bihlo&Popovych:eq:RestrictedFormOfPointTransForBVE} takes into account all constraints on point symmetries of~\eqref{Bihlo&Popovych:eq:vortbeta}, which follow from the preservation of the maximal Lie invariance algebra~$\mathfrak g$ by the associated push-forward of vector fields.

Now the direct method should be applied.
We carry out a transformation of the form~\eqref{Bihlo&Popovych:eq:RestrictedFormOfPointTransForBVE} in the vorticity equation.
For this aim, we calculate the transformation rules for the partial derivative operators:
\[
 \partial_{\tilde t}=\frac{1}{T_1}\left(\partial_t-\frac{f_t}{X_1}\partial_x\right), \quad\partial_{\tilde x}=\frac{1}{X_1}\partial_x, \quad \partial_{\tilde y}= \frac{1}{Y_1}\left(\partial_y-\frac{X_2}{X_1}\partial_x\right).
\]

Further restrictions on $\mathcal T$ can be imposed upon noting that the term $\psi_{txy}$ can only arise in the expression for $\tilde \psi_{\tilde t \tilde y\tilde y}$, which is
\[
 \tilde \psi_{\tilde t\tilde y\tilde y}=-\frac{2\Psi_1}{T_1Y_1}\frac{X_2}{X_1}\psi_{txy}+\cdots.
\]
This obviously implies that $X_2=0$. In a similar fashion, the expression for $\tilde \zeta_{\tilde t}$ is
\[
 \tilde \zeta_{\tilde t}=\frac{\Psi_1}{T_1}\left(\frac{1}{(X_1)^2}\zeta_t+\left(\frac{1}{(Y_1)^2}-\frac{1}{(X_1)^2}\right)\psi_{yyt}\right)+\cdots,
\]
upon using $\psi_{xxt}=\zeta_t-\psi_{yyt}$. Hence $(X_1)^2=(Y_1)^2$ as there are no other terms with $\psi_{yyt}$ in the invariance condition. After taking into account these two more restrictions on $\mathcal T$, it is straightforward to expand the transformed version of the vorticity equation. This yields
\begin{align*}
  &\frac{\Psi_1}{T_1(X_1)^2}\zeta_t-\frac{f_t\Psi_1}{T_1(X_1)^3}\zeta_x+\frac{(\Psi_1)^2}{(X_1)^3Y_1}\psi_x\zeta_y-\left(\frac{\Psi_1}{Y_1}\psi_y+ \frac{\Psi_2}{Y_1}\right)\frac{\Psi_1}{(X_1)^3}\zeta_x {} \\
 & {}+\beta\frac{\Psi_1}{X_1}\psi_x=\frac{\Psi_1}{T_1(X_1)^2}\left(\zeta_t+\psi_x\zeta_y-\psi_y\zeta_x+\beta \psi_x\right).
\end{align*}
The invariance condition is fulfilled provided that the constraints
\[
 \Psi_2=-\frac{Y_1}{T_1}f_t, \quad X_1=T_1(X_1)^2, \quad \frac{(\Psi_1)^2}{(X_1)^3Y_1}=\frac{\Psi_1}{T_1(X_1)^2}.
\]
hold. This completes the proof of the theorem.
\end{proof}

\begin{corollary}
 The barotropic vorticity equation on the $\beta$-plane possesses only two independent discrete point symmetries, which are given by
\[
 \Gamma_1\colon (t,x,y,\psi)\mapsto (-t,-x,y,\psi), \quad  \Gamma_2\colon (t,x,y,\psi)\mapsto (t,x,-y,-\psi).
\]
They generate the group of discrete symmetry transformations of the barotropic vorticity equation on the $\beta$-plane, which is isomorphic to $\mathbb Z^2\times\mathbb Z^2$, where $\mathbb Z^2$ denotes the cyclic group of two elements.
\end{corollary}

\section{Direct method and admissible transformations\\ of classes of generalized vorticity equations}\label{Bihlo&Popovych:sec:DirectMethodForBVE}

The construction of the complete point symmetry group~$G$ of the barotropic vorticity equation~\eqref{Bihlo&Popovych:eq:vortbeta}
by means of using only the direct method involves cumbersome and sophisticated calculations.
As Eq.~\eqref{Bihlo&Popovych:eq:vortbeta} is a third-order PDE in three independent variables,
the system of determining equations for transformations from~$G$ is an overdetermined nonlinear system of PDEs in four independent variables,
which should be solved by taking into account the nonsingularity condition of the point transformations.
This is an extremely challenging task.
Fortunately, a hierarchy of normalized classes of generalized vorticity equations
was recently constructed~\cite{Bihlo&Popovych:Popovych&Bihlo2010} that allows us to strongly simplify the whole investigation.
Eq.~\eqref{Bihlo&Popovych:eq:vortbeta} belongs only to the narrowest class of this hierarchy,
which is quite wide and consists of equations of the general form
\begin{gather}\label{Bihlo&Popovych:eq:class1}
\zeta_t=F(t,x,y,\psi,\psi_x,\psi_y,\zeta,\zeta_x,\zeta_y,\zeta_{xx},\zeta_{xy}, \zeta_{yy}), \quad \zeta:=\psi_{xx}+\psi_{yy},
\end{gather}
where $(F_{\zeta_x},F_{\zeta_y},F_{\zeta_{xx}},F_{\zeta_{xy}},F_{\zeta_{yy}})\ne(0,0,0,0,0)$.
The equivalence group~$G^\sim_1$ of this class is formed by the transformations
\begin{gather*}
\tilde t=T(t), \quad \tilde x=Z^1(t,x,y), \quad \tilde y=Z^2(t,x,y), \quad \tilde\psi=\Upsilon(t)\psi+\Phi(t,x,y), \\
\tilde F=\frac1{T_t}\left(
\frac\Upsilon LF+\Bigl(\frac\Upsilon L\Bigr)_0\zeta+\Bigl(\frac{\Phi_{ii}}L\Bigr)_0
-\frac{Z^i_tZ^i_j}L\left(\frac\Upsilon L\zeta_j+\Bigl(\frac\Upsilon L\Bigr)_j\zeta+\Bigl(\frac{\Phi_{ii}}L\Bigr)_j\right)
\right),
\end{gather*}
where $T$, $Z^i$, $\Upsilon$ and $\Phi$ are arbitrary smooth functions of their arguments,
satisfying the conditions $Z^1_kZ^2_k=0$, $Z^1_kZ^1_k=Z^2_kZ^2_k:=L$ and $T_t\Upsilon L\ne0$.
The subscripts~1 and~2 denote differentiation with respect to~$x$ and~$y$, respectively,
the indices~$i$ and~$j$ run through the set $\{1,2\}$ and the summation over repeated indices is understood.
As Eq.~\eqref{Bihlo&Popovych:eq:vortbeta} is an element of the class~\eqref{Bihlo&Popovych:eq:class1} and this class is normalized,
the point symmetry group~$G$ of Eq.~\eqref{Bihlo&Popovych:eq:vortbeta} is contained in the projection~$\hat G^\sim_1$ of the equivalence group~$G^\sim_1$
of the class~\eqref{Bihlo&Popovych:eq:class1} to the variable space $(t,x,y,\psi)$.
At the same time, the group~$G$ is much narrower than the group~$\hat G^\sim_1$,
and in order to single out~$G$ from~$\hat G^\sim_1$ we should still derive and solve a quite cumbersome system of additional constraints.
Instead of this we use the trick described in the end of Section~\ref{Bihlo&Popovych:sec:DirectMethodAndAdmTrans}. Namely, by the transformation
\begin{equation}\label{Bihlo&Popovych:eq:TrickTransOfPsi}
\check\psi=\psi+\frac\beta6y^3,
\end{equation}
which identically acts on the independent variables and which is prolonged to the vorticity according to the formula~$\check\zeta=\zeta+\beta y$,
we map Eq.~\eqref{Bihlo&Popovych:eq:vortbeta} to the equation
\begin{align}\label{Bihlo&Popovych:eq:vortbetaMod}
\check\zeta_t+\check\psi_x\check\zeta_y-\check\psi_y\check\zeta_x=-\frac\beta2y^2\check\zeta_x.
\end{align}
Eq.~\eqref{Bihlo&Popovych:eq:vortbetaMod} belongs to the subclass of class~\eqref{Bihlo&Popovych:eq:class1} that is
singled out by the constraints $F_\psi=0$, $F_\zeta=0$, $F_{\psi_x}=-\zeta_y$ and $F_{\psi_y}=\zeta_x$,
i.e., the class consisting of the equations of the form
\begin{equation}\label{Bihlo&Popovych:eq:class2}
   \zeta_t+\psi_x\zeta_y-\psi_y\zeta_x=H(t,x,y,\zeta_x,\zeta_y,\zeta_{xx},\zeta_{xy}, \zeta_{yy}), \quad \zeta :=\psi_{xx}+\psi_{yy},
\end{equation}
where $H$ is an arbitrary smooth function of its arguments, which is assumed as an arbitrary element instead of $F=H-\psi_x\zeta_y+\psi_y\zeta_x$.
The class~\eqref{Bihlo&Popovych:eq:class2} also is a member of the above hierarchy of normalized classes.
Its equivalence group~$G^\sim_2$ is much narrower than~$G^\sim_1$ and is formed by the transformations
\[%\label{EqTransFromGequiv2}
\begin{split}
&\tilde t=\tau, \quad
\tilde x=\lambda(x\mathfrak c-y\mathfrak s)+\gamma^1, \quad
\varepsilon\tilde y=\lambda(x\mathfrak s+y\mathfrak c)+\gamma^2, \\
&\tilde\psi=\varepsilon\frac{\lambda}{\tau_t}\left(\lambda\psi+\frac\lambda2\theta_t(x^2{+}y^2)
-\gamma^1_t(x\mathfrak s{+}y\mathfrak c)+\gamma^2_t(x\mathfrak c{-}y\mathfrak s)\right)+\delta+\frac\sigma2(x^2{+}y^2),\\
&\tilde H=\frac\varepsilon{\tau_t{}^2}
\left(H-\frac{\lambda_t}\lambda(x\zeta_x+y\zeta_y)+2\theta_{tt}\right)
-\frac{\delta_y{+}\sigma y}{\tau_t\lambda^2}\zeta_x+\frac{\delta_x{+}\sigma x}{\tau_t\lambda^2}\zeta_y
+\frac2{\tau_t}\left(\frac\sigma{\lambda^2}\right)_t,
\end{split}
\]
where $\varepsilon=\pm1$, $\mathfrak c=\cos\theta$, $\mathfrak s=\sin\theta$;
$\tau$, $\lambda$, $\theta$, $\gamma^i$ and $\sigma$ are arbitrary smooth functions of~$t$
satisfying the conditions $\lambda>0$, $\tau_{tt}=0$ and $\tau_t\ne0$ and
$\delta=\delta(t,x,y)$ runs through the set of solutions of the Laplace equation $\delta_{xx}+\delta_{yy}=0$.

In order to derive the additional constraints that are satisfied by the group parameters of transformations from
the point symmetry group~$G_2$ of Eq.~\eqref{Bihlo&Popovych:eq:vortbetaMod}, we substitute the values
$H=-\beta y^2\zeta_x/2$ and~$\tilde H=-\beta\tilde y^2\tilde \zeta_{\tilde x}/2$ as well as
expressions for the transformed variables and derivatives via the initial ones into the transformation component for~$H$
and then make all possible splitting in the obtained equality.
As a result, we derive the additional constraints
\[
\theta=\gamma^2_t=0, \quad \lambda=\frac1{\tau_t}, \quad \sigma=\frac{\varepsilon\beta\gamma^2}{2\tau_t{}^2}, \quad
\delta_x=-\sigma x, \quad \delta_y=\sigma y+\frac{\varepsilon\beta(\gamma^2)^2}{2\tau_t}.
\]
After projecting transformations from~$G^\sim_2$ on the variable space $(t,x,y,\psi)$,
constraining the group parameters using the above conditions
and taking the adjoint action of the inverse of the transformation~\eqref{Bihlo&Popovych:eq:TrickTransOfPsi},
we obtain, up to re-denoting, the transformations from Theorem~\ref{Bihlo&Popovych:TheoremOnPointSymGroupOfBVE}.

\section{Conclusion}\label{Bihlo&Popovych:sec:Conclusion}

In this paper, we have computed the complete point symmetry group of the barotropic vorticity equation on the $\beta$-plane. It is obvious that both of the techniques presented in this paper are applicable to general systems of differential equations.

Despite of the apparent simplicity of the techniques employed above, there are a number of features that should be discussed properly. In particular, the relation between discrete symmetries of a differential equation and discrete automorphisms of the corresponding maximal Lie invariance algebra is neither injective nor surjective. This is why it can be misleading to restrict the consideration to discrete automorphism when trying to finding discrete symmetries. This and related issues will be investigated and discussed more thoroughly in a forthcoming work.

\subsection*{Acknowledgements}

AB is a recipient of a DOC-fellowship of the Austrian Academy of Sciences. The research of ROP was supported by the project P20632 of the Austrian Science Fund.

\LastPageEnding

\end{document}